\documentclass[a4paper,UKenglish,cleveref, autoref, thm-restate]{lipics-v2019}

\title{PACE Solver Description:\\Bute-Plus: A Bottom-Up Exact Solver for Treedepth}

\titlerunning{Bute-Plus: A Bottom-Up Exact Solver for Treedepth}

\author{James Trimble}{School of Computing Science, University of Glasgow \\ Glasgow, Scotland, UK }{j.trimble.1@research.gla.ac.uk}{https://orcid.org/0000-0001-7282-8745}{This work was supported by the Engineering and Physical Sciences Research Council (grant number EP/R513222/1).}

\authorrunning{J. Trimble} 

\relatedversion{\url{https://drops.dagstuhl.de/opus/volltexte/2020/13337/}}

\category{PACE Solver Description}

\Copyright{James Trimble} 



\supplement{DOI of code submitted to PACE Challenge: \url{https://doi.org/10.5281/zenodo.3881441}.
  The latest version of the code can be found on GitHub: \url{https://github.com/jamestrimble/pace2020-treedepth-solvers}.}

\acknowledgements{I would like to thank the PACE 2020 program committee and the optil.io team for
a well-organised and enjoyable contest.}

\nolinenumbers 

\hideLIPIcs  

\EventEditors{Yixin Cao and Marcin Pilipczuk}
\EventNoEds{2}
\EventLongTitle{15th International Symposium on Parameterized and Exact Computation (IPEC 2020)}
\EventShortTitle{IPEC 2020}
\EventAcronym{IPEC}
\EventYear{2020}
\EventDate{December 14--18, 2020}
\EventLocation{Hong Kong, China (Virtual Conference)}
\EventLogo{}
\SeriesVolume{180}
\ArticleNo{34}

\usepackage{tikz}

\newcommand{\calS}{\mathcal{S}}

\begin{document}

\maketitle

\begin{abstract}
This note introduces \emph{Bute-Plus}, an exact solver for the treedepth
problem. The core of the solver is a positive-instance driven dynamic program that constructs
an elimination tree of minimum depth in a bottom-up fashion.
Three features greatly improve the algorithm's run time.  The first of
these is a specialised trie data
structure.  The second is a domination rule.
The third is a heuristic presolve
step can quickly find a treedepth decomposition of optimal depth for
many instances.
\end{abstract}

\section{Introduction}\label{sec:introduction}

A treedepth decomposition of graph $G$ is a rooted forest $F$, such that if
$G$ has edge
$\{u,v\}$ then either $u$ is an ancestor of $v$ or $v$ is
an ancestor of $u$ in $F$.  The treedepth problem is to determine, for a given
graph $G$, the minimum depth of a treedepth decomposition of $G$, where depth
is defined as the maximum number of vertices on a root-leaf path.

An \emph{elimination tree} of a connected graph $G$ is a special type of treedepth
decomposition, defined recursively as follows.  If $G$ has
a single vertex, its elimination tree equals $G$.
Otherwise, let $v$ be a vertex in $G$ and let $F$ be
a forest consisting of an elimination tree for each component of $G - v$.
Then an elimination tree of $G$ is formed by making $v$ the parent
of every root of $F$.

For every connected graph $G$, there exists an elimination tree whose
depth equals the treedepth of $G$
(\cite{DBLP:books/daglib/0030491}, chapter 6).  To solve the treedepth problem,
it is therefore sufficient to find an elimination tree of minimum depth.
That is the approach taken by the Bute-Plus solver, which this paper
introduces.
The solver uses a positive-instance driven dynamic programming algorithm,
which seeks sets of vertices that induce low-treedepth subgraphs of the input graph.
Three additional features improve the performance of the algorithm:
a specialised trie data structure, a domination
rule based on a rule by Ganian et al.~\cite{DBLP:conf/alenex/GanianLOS19}, and a heuristic presolver
that quickly finds an optimal solution for many of the PACE Challenge instances.

The author submitted two other exact solvers to the PACE Challenge: Bute
(which is Bute-Plus without the heuristic presolve step) and Bute-Plus-Plus
(which spends additional time on the heuristic presolve and has a minor modification
to the trie data structure).
An earlier algorithm by the author \cite{DBLP:conf/wea/000120a},
which constructs a treedepth decomposition from
the top down,
is very memory-efficient
but is typically much slower than Bute-Plus.

\section{A brief description of the algorithm}\label{sec:algorithm}

This section presents an outline of the Bute-Plus algorithm.
We assume that the vertex set of a graph $G$, denoted $V(G)$, contains only integers.
The neighbourhood of vertex $v$ is denoted by $N(v)$.
For a set of vertices $S$, $N(S)$ denotes the set of vertices
that are not in $S$ but are adjacent to some member of $S$.

The algorithm takes as input a connected graph $G$ and returns the treedepth of $G$
along with a treedepth decomposition of that depth.
The optimisation problem is solved as a sequence of decision problems.  The
solver attempts to find an elimination tree of depth 1, then of depth 2, and so
on until it is successful.  (Typically, the higher-numbered decision problems
are by far the most time consuming.)

For the decision problem of whether an elimination tree of depth $k$ exists,
the algorithm works downwards for $i = k, \dots, 1$, finding all subsets $S$ of $V(G)$
such that (1) $S$ induces a subgraph of $G$ with treedepth no greater than $k-i+1$, (2) the neighbourhood
of $S$ has fewer than $i$ vertices, and (3) the subgraph of $G$ induced by $S$
is connected.  This collection of sets of vertices is called $\calS_i^k$; it includes the vertex set
of every subtree whose root is at depth $i$ of an elimination tree of depth $k$ of $G$.

The algorithm uses a positive-instance driven (PID) \cite{DBLP:journals/jco/Tamaki19} approach to constructing
the $\calS_i^k$: rather than generating all subsets of $V(G)$
and checking if each one satisifies the three required properties, the elements
of $\calS_i^k$ are generated by joining together elements of $\calS_{i+1}^k$.
To be more precise, sets in $\calS_i^k$ are constructed
in two ways; a sketch of these follows.  The first
is simply by choosing vertices with a sufficiently small neighbourhood (since clearly
each of these induces a connected subgraph of treedepth 1).  The second
is by finding a nonempty sub-collection $\calS \subseteq \calS_{i+1}^k$
and a vertex $v$ satisfying the following conditions.
The elements of $\calS$ must be pairwise disjoint, and moreover there must not be an edge
between vertices in any two distinct members of $\calS$.  Furthermore, $v$ must have an edge
to at least one vertex in each member of $\calS$.  These conditions guarantee that
the set $\bigcup{\calS} \cup \{v\}$ induces a connected subgraph of $G$ of that admits an elimination
tree with root $v$ of depth no more than $k-i+1$.

It is easy to verify using the definition of $\calS_i^k$ that an instance of the decision
problem is satisfiable if and only if $\calS_1^k$ is non-empty (in which case $\calS_1^k$
will have $V(G)$ as its only element).  A small amount of extra bookkeeping allows the solver to output an optimal elimination
tree.

Bute-Plus is not the first PID algorithm for treedepth.
Bannach and Berndt \cite{DBLP:conf/wads/BannachB19} present a PID framework
for computing a range of graph parameters including treedepth, treewidth, and pathwidth.
Although their paper describes the family of algorithms in terms of a game theoretic characterisation
of each problem, their algorithm for treedepth has a similar overall approach
to that of Bute-Plus: both algorithms build up sets of vertices by combining
one or more existing sets with a root vertex.
Bannach and Berndt use a queue when combining sets whereas Bute-Plus uses a stack;
a second difference is that the algorithm of Bannach and Berndt is not
restricted to finding only elimination trees.  The framework of Bannach and Berndt
generalises a PID algorithm for treewidth by Tamaki which won the exact treewidth track
of PACE 2016;\footnote{\url{https://github.com/TCS-Meiji/treewidth-exact}}
a second PID algorithm for treewidth by Tamaki
\cite{DBLP:journals/jco/Tamaki19}
performed strongly in PACE 2017.

\section{Improvements to the algorithm}

The Bute-Plus solver has three additional features which greatly reduce run time
on many instances.  Two of these---a trie data structure and a domination rule---are
described in the following two subsections.  The third feature is a heuristic solver
which is run for the first minute with the hope of finding a treedepth decomposition 
of optimal depth; this uses the Tweed-Plus solver which was an entry by the author
in the heuristic
track of PACE 2020 and is described in its own paper in this volume.

\subsection{Trie data structure}

Recall from \Cref{sec:algorithm} that the algorithm generates sets in the collection
$\calS_i^k$
by finding a subset of
$\calS_{i+1}^k$ along with a vertex $v$ that together satisfy certain properties.
For some of the PACE Challenge instances, $\calS_{i+1}^k$ can contain millions of
sets, and the task of finding appropriate subsets of the collection
becomes intractable without a specialised data structure.

Bute-Plus's data structure supports two operations.  The first is to add
a $(S, N(S))$ pair---a set of vertices and its neighbourhood---to the collection.
The second is a query operation which takes a set of vertices $Q$ and an integer
$i$.  This returns all sets $S$ in the collection such that
both (1) $|N(S) \cup N(Q)| < i$ and (2) $(Q \cup N(Q)) \cap S = \emptyset$.

The data structure is implemented as a trie.
When $(S, N(S))$ is inserted, $N(S)$ is sorted in ascending order and
viewed as a string over the alphabet $V(G)$, then added to the trie.  This
approach has been used for the similar problem of superset queries
several times in the past, for example in Savnik's Set-Trie
\cite{DBLP:conf/IEEEares/Savnik13}.  

To sketch the query operation: the algorithm performs a depth-first traversal of
the trie, backtracking when it becomes clear that no value in the subtree is
acceptable.  For efficiency, each node of the trie stores the intersection of
$N(S)$ values in the subtree rooted at that node; this idea is from
a data structure for superset queries posted on Stack Overflow by Ben Tilly
\cite{TrieStackOverflow}.

The task carried out by Bute-Plus's data structure
is similar to the task of the \emph{block sieve} designed by Tamaki for a
PID treewidth solver \cite{DBLP:journals/jco/Tamaki19}.
Although both data structures are based on tries, their designs differ in
several respects; for example, the block sieve data structure comprises
a collection of tries rather than just one.

\subsection{Domination rule}

As discussed in \Cref{sec:introduction}, to find a minimum-depth treedepth
decomposition it is sufficient to restrict attention to elimination trees.
We can speed the algorithm up further by placing additional restrictions
on acceptable elimination trees, if it can be shown that at least one tree
in the restricted class has optimal depth.

For this purpose, the Bute algorithm uses a domination-breaking rule that extends a rule
by Ganian et al.\ \cite{DBLP:conf/alenex/GanianLOS19}.
For distinct vertices $v, w$, we say that $v$ \emph{dominates} $w$
if either of the following two conditions holds: (1)
$N(v) \setminus \{w\} \supset N(w) \setminus \{v\}$;
(2) $N(v) \setminus \{w\} = N(w) \setminus \{v\}$
and $w < v$.  It is always possible to construct an elimination tree of minimum
depth such that no vertex dominates any of its ancestors.

This rule allows us to further restrict each collection $\calS_i^k$ to include
only sets of vertices $S$ such that no vertex in $S$ dominates any member of $N(S)$.

\section{Implementation details}

The Bute-Plus solver is written in C.  Sets of vertices are stored using bitsets;
code from Nauty 2.6r12
\cite{DBLP:journals/jsc/McKayP14} is used for the bitset data structure.\footnote{Nauty is available at \url{http://pallini.di.uniroma1.it/}}
The Tweed-Plus heuristic presolver also uses code from Nauty for
for the random number generator, and uses Metis 5.1.0 \cite{DBLP:journals/siamsc/KarypisK98}
to find nested dissection orderings.\footnote{Metis
is available at \url{http://glaros.dtc.umn.edu/gkhome/metis/metis/overview}}

\bibliography{bib}

\appendix
\section{Proof of correctness for domination rule} \label{appendix:proof}

We now prove the correctness of the domination rule.
Given a graph $G$ and distinct vertices $v, w \in V(G)$
(with no restriction on whether $v$ and $w$ are adjacent),
we say that $v$ \emph{dominates} $w$
in $G$ if either of the following conditions holds:

\begin{itemize}
  \item $N(v) \setminus \{w\} \supset N(w) \setminus \{v\}$; or
  \item $N(v) \setminus \{w\} = N(w) \setminus \{v\}$ and $v > w$.
\end{itemize}

We use the notation $G[S]$ to denote the subgraph of $G$ induced by
vertex-set $S$, and $T[v]$ to denote the subtree of $T$ rooted at vertex
$v$.

\begin{theorem}
Let a connected graph $G$ be given.  There exists an elimination
tree $T$ of $G$ whose depth equals the treedepth of $G$, such
that for every vertex $w \in V(G)$ and every vertex $v \in V(G)$ that
dominates $w$ in $G$ we have that $w$ is not an ancestor of $v$ in $T$.
\end{theorem}

\begin{proof}
For ease
of exposition, we assume that the vertices of $G$ are numbered in nondecreasing
order of degree
(i.e. $v < w \implies |N(v)| \leq |N(w)|$), but the proof can easily be generalised
by defining an appropriate ordering relation on the vertices.

For an elimination tree $T$ of $G$, we define the \emph{score}
function $s_T : V(G) \mapsto \mathbb N \times V(G)$ that maps each
vertex $v$ to the tuple $(d,v)$ where $d$ is the depth of $v$ in $T$.
We compare scores lexicographically; thus, $v$ has a higher score than $v'$
if $v$ appears deeper in the tree than $v'$ or if the two vertices are at
the same depth and $v > v'$.  We also define the score of a tree: the score
of $T$ equals the minimum $s_T(v)$ over all vertices $v$ that are dominated
by one of their descendants.  If no such $v$ exists, the score of $T$
is the special value $(\infty, \infty)$.

Let a minimum-depth elimination tree $T$ of $G$ that breaks the domination
rule be given; that
is, there exist $v,v' \in V(G)$ such that $v$ dominates $v'$ in $G$ and
$v'$ is an ancestor of $v$ in $T$.
We will demonstrate that it is possible to reorder a subtree of $T$ to obtain
a new minimum-depth elimination tree of strictly greater score than $T$.  Repeated
application of this rule must yield a minimum-depth elimination tree that
does not break the domination rule, since there are only finitely many
different scores that elimination trees of a finite graph can have.

Let $(d,u)$ be the score of $T$.
Let $w$ be the greatest-numbered vertex in $T[u]$ that dominates $u$.
The subtree $T[u]$ may be replaced with an elimination tree that has the
same vertex set as $T[u]$ but is rooted at $w$, without increasing the height of the
subtree (since $G[V(T[u]) \setminus \{w\}]$
is isomorphic to a subgraph of $G[V(T[u]) \setminus \{u\}]$).
This replacement results in new minimum-depth elimination tree $T'$ of
$G$. The score of $T'$ is at least $(d, w)$, which is greater than $(d,u)$.
\end{proof}

\end{document}